\documentclass[a4paper]{article}

\usepackage{amsmath, amssymb, amsthm, bbm, booktabs, enumitem, graphicx, url, color} 
\usepackage[a4paper]{geometry}
\usepackage[authoryear, round]{natbib}

\usepackage[breaklinks, colorlinks=true, citecolor=black, urlcolor=black, linkcolor=black]{hyperref}

\newcommand{\N}{\mathbb{N}}
\newcommand{\R}{\mathbb{R}}
\newcommand{\E}{\mathbb{E}}
\renewcommand{\P}{\mathbb{P}}

\newcommand{\bX}{\mathbf{X}}

\newcommand{\cA}{\mathcal{A}}
\newcommand{\cL}{\mathcal{L}}
\newcommand{\cT}{\mathcal{T}}
\newcommand{\cX}{\mathcal{X}}
\newcommand{\cY}{\mathcal{Y}}
\newcommand{\cPY}{\mathcal{P}(\mathcal{Y})}

\newcommand{\myE}{\mathbb{E}}
\newcommand{\myQ}{\mathbb{P}}
\newcommand{\myS}{\operatorname{S}}
\newcommand{\myT}{\operatorname{T}}

\newcommand{\Deltam}{\Delta_{m-1}}
\newcommand{\Sym}{\operatorname{Sym}}

\newcommand{\one}{\mathbbm{1}}

\DeclareMathOperator*{\argmin}{arg\,min}
\DeclareMathOperator*{\argmax}{arg\,max}

\theoremstyle{plain}

\newtheorem{Thm}{Theorem}[section]
\newtheorem{Prop}[Thm]{Proposition}

\newtheorem{Lem}[Thm]{Lemma}

\newtheorem{theorem}[Thm]{Theorem}

\theoremstyle{definition}

\newtheorem{Ex}[Thm]{Example}

\newtheorem{example}[Thm]{Example}

\begin{document}

\title{From Classification Accuracy 
		to 
		Proper Scoring Rules: Elicitability of Probabilistic Top List Predictions} 
\author{Johannes Resin\thanks{This work has been supported by the Klaus Tschira Foundation. The author gratefully acknowledges financial support from the German Research Foundation (DFG) through grant number 502572912. The author would like to thank Timo Dimitriadis, Tobias Fissler, Tilmann Gneiting, Alexander Jordan, Sebastian Lerch and Fabian Ruoff for helpful comments and discussion.}\\[2mm]
	Heidelberg University \\
	Heidelberg Institute for Theoretical Studies} 

\maketitle

\begin{abstract} \noindent
		In the face of uncertainty, the need for probabilistic assessments has long been recognized in the literature on forecasting. In classification, however, comparative evaluation of classifiers often focuses on predictions specifying a single class through the use of simple accuracy measures, which disregard any probabilistic uncertainty quantification. I propose probabilistic top lists as a novel type of prediction in classification, which bridges the gap between single-class predictions and predictive distributions. The probabilistic top list functional is elicitable through the use of strictly consistent evaluation metrics. The proposed evaluation metrics are based on symmetric proper scoring rules and admit comparison of various types of predictions ranging from single-class point predictions to fully specified predictive distributions. The Brier score yields a metric that is particularly well suited for this kind of comparison.
\end{abstract}

\section{Introduction}

In the face of uncertainty, predictions ought to quantify their level of confidence \citep{GK14}. This has been recognized for decades in the literature on weather forecasting \citep{Brier1950, Murphy1977} and probabilistic forecasting \citep{Daw84,GR07}. Ideally, a prediction specifies a probability distribution over potential outcomes. Such predictions are evaluated and compared by means of proper scoring rules, which quantify their value in a way that rewards truthful prediction \citep{GR07}. In statistical classification and machine learning, the need for reliable uncertainty quantification has not gone unnoticed, as exemplified by the growing interest in the calibration of probabilistic classifiers \citep{GPSW17,VWA+19}. However, classifier evaluation often focuses on the most likely class (i.e., the mode of the predictive distribution) through the use of classification accuracy and related metrics derived from the confusion matrix \citep{Tha20,HB21}.

Probabilistic classification separates the prediction task from decision making. This enables informed decisions that account for diverse cost-loss structures, for which decisions based simply on the most likely class may lead to adverse outcomes \citep{Elkan01,G17}. Probabilistic classification 
is a viable alternative to classification with reject option, where classifiers may refuse to predict a class if their confidence in a single class is not sufficient \citep{HW06,NCHS19}.

In this paper, I propose \emph{probabilistic top lists} as a way of producing probabilistic classifications in settings where specifying entire predictive distributions may be undesirable, impractical or even impossible. While multi-label classification serves as a key example of such a setting, the theory presented here applies to classification in general. I envision the probabilistic top list approach to be particularly useful in settings eluding traditional probabilistic forecasting, where the specification of probability distributions on the full set of classes is hindered by a large number of classes and missing (total) order. Consistent evaluation is achieved through the use of proper scoring rules. 

Whereas in traditional classification an instance is associated with a single class (e.g., \emph{cat} or \emph{dog}), multi-label classification problems \citep[as reviewed by][]{TK07,ZZ14,TGM21} admit multiple labels for an instance (e.g., \emph{cat} or \emph{dog} or \emph{cat and dog}).\footnote{Multi-label classification is a special case of classification if classes are (re-)defined as subsets of labels.}
Applications of multi-label classification include text categorization \citep{ZZ06}, image recognition \citep{CWWG19} and functional genomics \citep{BST06,ZZ06}. Multi-label classification methods often output confidence scores for each label independently and the final label set prediction is determined by a simple cut-off \citep{ZZ14}. As this does not take into account label correlations, computing label set probabilities in a postprocessing step can improve predictions and probability estimates \citep{LPAWQ20} over simply multiplying probabilities to obtain label set probabilities.
Probabilistic top lists offer a flexible approach to multi-label classification, which embraces the value of probabilistic information. In fact, the \emph{BR-rerank} method introduced by \citet{LPAWQ20} produces top list predictions. Yet, comparative performance evaluation focuses on (set) accuracy and the improper instance F1 score. This discrepancy has been a key motivation for this research.

In probabilistic forecasting, a scoring rule assigns a numerical score to a predictive distribution based on the true outcome \citep{GR07}. It is proper if the expected score is optimized by the true distribution of the outcome of interest. Popular examples in classification are the Brier (or quadratic) score and the logarithmic (or cross entropy) loss \citep{GR07,HB21}. When one is not interested in full predictive distributions, simple point predictions are frequently preferred. A meaningful point prediction admits interpretation in terms of a statistical functional \citep{Gneiting2011}. Point predictions are evaluated by means of consistent scoring or loss functions. Similar to proper scoring rules, a scoring function is consistent for a functional if the expected score is optimized by the true functional value of the underlying distribution. For example, accuracy (or, equivalently, misclassification or zero-one loss) is consistent for the mode in classification \citep{G17}.

Probabilistic top lists bridge the gap between mode forecasts and full predictive distributions in classification. In this paper, I define a probabilistic top-$k$ list as a collection of $k$ classes deemed most likely together with confidence scores quantifying the predictive probability associated with each of the $k$ classes. The key question tackled in this work is how to evaluate such top list predictions in a consistent manner. To this end, I propose what I call \emph{padded symmetric scores}, which are based on proper symmetric scoring rules. 
I show that the proposed padded symmetric scores are consistent for the probabilistic top-$k$ list functional. The padded symmetric score of a probabilistic top list prediction is obtained from a symmetric proper scoring rule by padding the top list to obtain a fully specified distribution. The padded distribution divides the probability mass not accounted for by the top list's confidence scores equally among the classes that are not included in the list. Padded symmetric scores exhibit an interesting property, which allows for balanced comparison of top lists of different length, as well as single-class point predictions and predictive distributions. Notably, the expected score of a correctly specified top list only depends on the top list itself and is invariant to other aspects of the true distribution.
Comparability of top lists of differing length is ensured, as the expected score does not deteriorate upon increasing the length of the predicted top list. Nonetheless, if the scoring function is based on the Brier score, there is little incentive to provide unreasonably large top lists.
In the case of a single-class prediction, the padded version of the Brier score reduces to twice the misclassification loss. 
Hence, the padded Brier score essentially generalizes classification accuracy.

The remainder of the paper proceeds as follows. Section \ref{Sec:StatClass} recalls the traditional multi-class classification problem with a focus on probabilistic classification and suitable evaluation metrics. A short introduction to the multi-label classification problem is also provided. Section \ref{Sec:TopLists} introduces probabilistic top lists and related notation and terminology used throughout this work. Section \ref{Sec:MathPrelim} introduces some preliminary results on symmetric proper scoring rules and some results relating to the theory of majorization. These results are used in Section \ref{Sec:PaddedScores} to show that the padded symmetric scores yield consistent scoring functions for the top list functionals. Section \ref{Sec:Comp} discusses the comparison of various types of predictions using the padded Brier and logarithmic scores. A theoretical argument as well as numerical examples illustrate that the padded Brier score is well suited for this task. Section \ref{Sec:ConclusionTopLists} concludes the paper.

\section{Statistical Classification}
\label{Sec:StatClass}

The top list functionals and the proposed scoring functions are motivated by multi-label classification, but they apply to other classification problems as well. Here, I give a short formal introduction to the general classification problem and related evaluation metrics from the perspective of probabilistic forecasting. In what follows, the symbol $\cL$ refers to the law or distribution of a given random variable.

\subsection{Traditional multi-class classification}
\label{Sec:Class}

In the classical (multi-class) classification problem, one tries to predict the distinct class $Y$ of an instance characterized by a vector of features $\bX$.  Formally, the outcome $Y$ is a random variable on a probability space $(\Omega,\cA,\myQ)$ taking values in the set of classes $\cY$ of cardinality $m \in \N$, and the feature vector $\bX$ is a random vector taking values in some feature space $\cX \subseteq \R^d$. Ideally, one learns the entire conditional distribution $p(\bX) = \cL(Y \mid \bX)$ of $Y$ given $\bX$ through a probabilistic classifier $c\colon\cX \rightarrow \cPY$ mapping the features of a given instance to a probability distribution from the set of probability distributions $\cPY$ on $\cY$. The set $\cPY$ of probability distributions is typically identified with the probability simplex
\[
\Delta_{m-1} = \{p \in [0,1]^m \mid p_1 + \dots + p_m = 1\}
\]
by (arbitrarily) labeling the classes as $1,\dots,m$, and probability distributions are represented by vectors $p \in \Delta_{m-1}$, where the $i$-th entry $p_i$ is the probability assigned to class $i$ for $i = 1,\dots,m$. To ease notation in what follows, vectors in $\Delta_{m-1}$ are indexed directly by the classes in $\cY$ without explicit mention of any (re-)labeling.

Proper scoring rules quantify the value of a probabilistic classification and facilitate comparison of multiple probabilistic classifiers \citep{GR07}. 
A scoring rule is a mapping $\myS \colon \cPY \times \cY \rightarrow \overline{\R}$, which assigns a, possibly infinite, score $\myS(p,y)$ from the extended real numbers $\overline\R = \R \cup \{\pm \infty\}$ to a predictive distribution $p$ if the true class is $y$. Typically, scores are negatively oriented in that lower scores are preferred.
A scoring rule $\myS$ is called \emph{proper} if the true distribution $p = \cL(Y)$ of $Y$ minimizes the expected score,
\begin{equation} \label{Eq:Proper2}
\E[\myS(p,Y)] \leq \E[\myS(q,Y)]\quad\text{for } Y\sim p \text{ and all } p,q \in \cPY.
\end{equation}
It is \emph{strictly proper} if the inequality \eqref{Eq:Proper2} is strict unless $p = q$. Prominent examples are the logarithmic score
\begin{equation}
\myS_{\log}(p,y) = -\log p_y
\label{Eq:LogS}
\end{equation}
and the Brier score
\begin{equation}
\myS_{\mathrm{B}}(p,y) = (1 - p_y)^2 + \sum_{z \neq y} p_z^2 = 1 - 2p_y + \sum_{z \in \cY} p_z^2.
\label{Eq:Brier}
\end{equation}

Frequently, current practice does not focus on learning the full conditional distribution, but rather on simply predicting the most likely class, i.e., the mode of the conditional distribution $p(\bX)$. This is formalized by a \emph{hard} classifier $c \colon \cX \rightarrow \cY$ aspiring to satisfy the functional relationship $c(\bX) \in \operatorname{Mode}(p(\bX))$, where the \emph{mode functional} is given by
\begin{equation}
\label{Eq:Mode}
\operatorname{Mode}(p) = \argmax_{y \in \cY} p_y = \{z \in \cY \mid p_z = \max_{y \in \cY} p_y\}
\end{equation}

for $p \in \Deltam$. Other functionals may be learned as well. When it comes to point forecasts of real-valued outcomes popular choices are the mean or a quantile, see for example \cite{GR21}.
Formally, a statistical functional $\myT\colon \cPY \rightarrow 2^{\cT}$ reduces probability measures to certain facets in some space $\cT$. Note that the functional $\myT$ maps a distribution to a subset in the power set $2^{\cT}$ of $\cT$ owing to the fact that the functional value may not be uniquely determined. For example, the mode \eqref{Eq:Mode} of a distribution is not unique if multiple classes are assigned the maximum probability.  The probabilistic top lists introduced in Section \ref{Sec:TopLists} are a nonstandard example of a statistical functional, which lies at the heart of this work.

Similar to the evaluation of probabilistic classifiers through the use of proper scoring rules, predictions aimed at a statistical functional are evaluated by means of consistent scoring functions. Given a functional $\myT$, a scoring function is a mapping $\myS\colon \cT \times \cY \rightarrow \overline{\R}$, which assigns a score $\myS(t,y)$ to a predicted facet $t$ if the true class is $y$. A scoring function $\myS$ is \emph{consistent} for the functional $\myT$ if the expected score is minimized by any prediction that is related to the true distribution of $Y$ by the functional, i.e., 
\begin{equation} \label{Eq:Consistent2}
\E[\myS(t,Y)] \leq \E[\myS(s,Y)]\quad\text{for } Y\sim p, t \in \myT(p) \text{ and all } p \in \cPY, s \in \cT.
\end{equation}
It is \emph{strictly consistent} for $\myT$ if the inequality \eqref{Eq:Consistent2} is strict unless $s \in \myT(p)$. A functional $\myT$ is called \emph{elicitable} if a strictly consistent scoring function for $\myT$ exists. For example, the mode \eqref{Eq:Mode} is elicited by the zero-one scoring function or misclassification loss \citep{G17}
\[
\myS(x,y) = \one\{x \neq y\},
\] 
which is simply a negatively oriented version of the ubiquitous classification accuracy. As discussed by \cite{G17} and references therein, decisions based on the mode are suboptimal if the losses invoked by different misclassifications are not uniform, which is frequently the case. 

(Strictly) Proper scoring rules arise as a special case of (strictly) consistent scoring functions if $\myT$ is the identity on $\cPY$. Furthermore, any consistent scoring function yields a proper scoring rule if predictive distributions are reduced by means of the respective functional first \citep[Theorem 3]{Gneiting2011}. On the other hand, a point prediction $x \in \cY$ can be assessed by means of a scoring rule, as the classes can be embedded in the probability simplex by identifying a class $y \in \cY$ with the point mass $\delta_y \in \cPY$ in $y$. For example, applying the Brier score to a class prediction in this way yields twice the misclassification loss, $\myS_{\mathrm{B}}(x,y) = \myS_{\mathrm{B}}(\delta_x,y) = 2\cdot \one\{x\neq y\}$.

Naturally, the true conditional distributions are unknown in practice and expected scores are estimated by the mean score attained across all instances available for evaluation purposes.

\subsection{Multi-label classification}

In multi-label classification problems, an instance may be assigned multiple (class) labels. Here, I frame this as a special case of multi-class classification instead of an entirely different problem.

Let $L$ be the set of labels and $\cY \subseteq 2^L$ be the set of label sets, i.e., classes are subsets of labels. In this setting, it may be difficult to specify a sensible predictive distribution on $\cY$ even for moderately sized sets of labels $L$, since the number of classes may grow exponentially in the number of labels. Extant comparative evaluation practices in multi-label classification focus mainly on hard classifiers ignoring the need for uncertainty quantification through probabilistic assessments \citep[e.g.,][]{TK07,ZZ14,LPAWQ20,TGM21} with the exception of \citet{RPHT11}, who also consider a sum of binary logarithmic losses to evaluate the confidence scores associated with individual labels.

Classification accuracy is typically referred to as (sub-)set accuracy in multi-label classification. Other popular evaluation metrics typically quantify the overlap between the predicted label set and the true label set. For example, the comparative evaluation by \cite{LPAWQ20} reports instance F1 scores in addition to set accuracy, where instance F1 of a single instance is defined as
\[
\myS_{\mathrm{F1}}(x,y) = \frac{2 \sum_{\ell \in L} \one\{\ell \in x\}\one\{\ell \in y\}}{\sum_{\ell \in L} \one\{\ell \in x\} + \sum_{\ell \in L} \one\{\ell \in y\}}.
\]
(and the overall score is simply the average across all instances as usual). Note that this is a positively oriented measure, i.e., higher instance F1 scores are preferred.
Caution is advised, as the instance F1 score is not consistent for the mode as illustrated by the following example. Hence, evaluating the same predictions using set accuracy and instance F1 seems to be a questionable practice.

\begin{example}
	Let the label set $L = \{1,2,3,4,5\}$ consist of five labels and the set of classes $\cY = 2^L$ be the power set of the label set $L$. Consider the distribution $p \in \cPY$ that assigns all probability mass to four label sets as follows:
	\[
	p_{\{1,2\}} = 0.28, \quad p_{\{1,3\}} = 0.24, \quad p_{\{1,4\}} = 0.24, \quad p_{\{1,5\}} = 0.24.
	\]
	Then the expected instance F1 score of the most likely label set $\{1,2\}$, 
	\[
	\myE[\myS_{\mathrm{F1}}(\{1,2\},Y)] = 0.64,
	\]
	given $Y \sim p$
	is surpassed by predicting only the single label $\{1\}$,
	\[
	\myE[\myS_{\mathrm{F1}}(\{1\},Y)] = \tfrac23.
	\] 
\end{example}

\section{Probabilistic Top Lists}
\label{Sec:TopLists}

In what follows, I develop a theory informing principled evaluation of top list predictions based on proper scoring rules. To this end, a concise mathematical definition of probabilistic top lists is fundamental.

Let $k \in \{0,\dots,m\}$ be fixed. A \emph{(probabilistic) top-$k$ list} is a collection $t = (\hat{Y}, \hat{t})$ of a set $\hat{Y} \subset \cY$ of $k = \vert\hat{Y}\vert$ classes together with a vector $\hat{t} = (\hat{t}_y)_{y \in \hat{Y}} \in [0,1]^k$ of \emph{confidence scores} (or predicted probabilities) indexed by the set $\hat{Y}$ whose sum does not exceed one, i.e., $\sum_{y \in \hat{Y}} \hat{t}_y \leq 1$, and equals one if $k = m$. Let $\cT_k$ denote the set of probabilistic top-$k$ lists. On the one hand, the above definition includes the empty top-0 list $t_\emptyset = (\emptyset,())$ for technical reasons. At the other extreme, top-$m$ lists specify entire probability distributions on $\cY$, i.e., $\cT_m \equiv \cPY$. The \emph{proxy probability}
\[
\pi(t) := \frac{1 - \sum_{y \in \hat{Y}} \hat{t}_y}{m - k}
\]
associated with a top-$k$ list $t = (\hat{Y},\hat{t}) \in \cT_k$ of size $k < m$ is the probability mass not accounted for by the top list $t$ divided by the number of classes not listed. For a top-$m$ list $t \in \cT_m$, the proxy probability $\pi(t) \equiv 0$ is defined to be zero.
The \emph{padded probability distribution} $\tilde{t} = (\tilde{t}_y)_{y \in \cY} \in \Deltam$ associated with a probabilistic top-$k$ list $t = (\hat{Y},\hat{t}) \in \cT_k$ assigns the proxy probability $\pi(t)$ to all classes not in $\hat{Y}$, i.e.,
\begin{equation} \label{Eq:PaddedDist}
\tilde{t}_y = \begin{cases}
\hat{t}_y, & \text{if } y \in \hat{Y}, \\
\pi(t), & \text{if } y \notin \hat{Y}
\end{cases}
\end{equation}
for $y \in \cY$.

A top-$k$ list $t = (\hat{Y},\hat{t})$ is \emph{calibrated} relative to a distribution $p = (p_y)_{y \in \cY} \in \Deltam$ if the confidence score $\hat{t}_y$ of class $y$ matches the true class probability $p_y$ for all $y \in \hat{Y}$. A top-$k$ list $t = (\hat{Y},\hat{t})$ is \emph{true} relative to a distribution $p \in \cPY$ if it is calibrated relative to $p$ and $\hat{Y}$ consists of $k$ most likely classes. There may be multiple true top-$k$ lists for a given $k \in \N$ if the class probabilities are not distinct (i.e., some classes have the same probability). References to the true distribution of the outcome $Y$ are usually omitted in what follows. For example, a calibrated top list is understood to be calibrated relative to the distribution $\cL(Y)$ of $Y$.
The \emph{(probabilistic) top-$k$ list functional} $\myT_k\colon \cPY \rightarrow \cT_k$ maps any probability distribution $p \in \cPY$ to the set
\[
\myT_k(p) = \left\{(\hat{Y},(p_y)_{y \in \hat{Y}}) \in \cT_k \;\middle|\; \hat{Y} \in \argmax_{S \subset \cY : \vert S\vert = k}\sum_{y \in S} p_y
\right\}
\]
of top-$k$ lists that are true relative to $p$.
The top-$m$ list functional $\myT_m$ identifies $\cPY$ with $\cT_m$. A top-$k$ list $t \in \cT_k$ is \emph{valid} if it is true relative to some probability distribution, i.e., there exists a distribution $p \in \cPY$ such that $t \in \myT_k(p)$. Equivalently, a top-$k$ list $t = (\hat{Y},\hat{t})$ is valid if the associated proxy probability does not exceed the least confidence score, i.e., $\min_{y \in\hat{Y}} \hat{t}_y \geq \pi(t)$. Let $\tilde{\cT}_k \subset \cT_k$ denote the set of valid top-$k$ lists.
The following is a simple example illustrating the previous definitions.

\begin{Ex}\sloppy
	Let $k = 2$, $m = 4$, $\cY = \{1,2,3,4\}$ and $Y \sim p = (0.5,0.2,0.2,0.1)$, i.e., $\P(Y = y) = p_y$. 	There are two true top-$2$ lists, namely, $T_2(p) = \{(\{1,2\},(0.5,0.2)),(\{1,3\},(0.5,0.2))\}$. The list $s = (\{1,4\},(0.5,0.1))$ is calibrated (relative to $p$), but fails to be valid, because it cannot be true relative to a probability distribution on $\cY$. On the other hand, the list $r = (\{1,4\},(0.5,0.2))$ is valid, as it is true relative to $q = (0.5,0.2,0.1,0.2)$, but fails to be calibrated.
	\label{Ex:TopLists}
\end{Ex}

An invalid top-$k$ list $t = (\hat{Y},\hat{t})$ contains a \emph{largest valid sublist} $t' = (\hat{Y}',(\hat{t}_y)_{y \in \hat{Y}'})$.
The largest valid sublist is uniquely determined by recursively removing the class $z \in \argmin_{y \in \hat{Y}} \hat{t}_y$ with the lowest confidence score from the invalid list until a valid list remains. Removing a class $x \in \hat{Y}$ with $\pi(t) > \hat{t}_x$ cannot result in a valid top list $t' = (\hat{Y}\setminus \{x\},(\hat{t}_y)_{y \in \hat{Y}\setminus \{x\}})$ as long as there is another class $z$ such that $\hat{t}_x \geq \hat{t}_z$, because $\pi(t) > \pi(t') > \hat{t}_x \geq \hat{t}_z$. Similarly, removing a class $x \in \hat{Y}$ with $\pi(t) \leq \hat{t}_x$ cannot prevent the removal of a class $z$ if $\pi(t) > \hat{t}_z$, because it does not decrease the proxy probability, $\pi(t') \geq p(t)$. Hence, \emph{no} sublist containing a class with minimal confidence score in the original list is valid and removal results in a superlist of the largest valid sublist. 

In what follows, I show how to construct consistent scoring functions for the top-$k$ list functional using proper scoring rules. 
Recall from Section \ref{Sec:Class} that a scoring function $\myS\colon \cT_k \times \cY \rightarrow \overline{\R}$ is \emph{consistent} for the top list functional $\myT_k$ if the expected score under any probability distribution $p \in \cPY$ is minimized by any true top-$k$ lists $t \in \myT_k(p)$, i.e.,
\[
\myE[\myS(t,Y)] \leq \myE[\myS(s,Y)]
\]
holds for $Y \sim p$ and any $s \in \cT_k$. 
It is \emph{strictly consistent} if the expected score is minimized only by the true top-$k$ lists $t \in \myT_k(p)$, i.e., the inequality is strict for $s \notin \myT_k(p)$. The functional $\myT_k$ is \emph{elicitable} if a strictly consistent scoring function for $\myT_k$ exists. In what follows, such a scoring function is constructed, giving rise to the following theorem.

\begin{Thm}
	The top-$k$ list functional $\myT_k$ is elicitable.
\end{Thm}

\begin{proof}
	This is an immediate consequence of either Theorem \ref{Thm:PaddedSym} or \ref{Thm:PenSym}.
\end{proof}

As the image of $\myT_k$ is $\tilde{\cT}_k$ by definition, invalid top-$k$ lists may be ruled out a priori and the domain of $\myS$ may be restricted to $\tilde{\cT}_k \times \cY$ in the above definitions. This is essentially a matter of taste and the question is whether predictions must be valid or whether this should merely be encouraged by the use of a consistent scoring function. Any scoring function that is consistent for valid top list predictions can be extended by assigning an infinite score to any invalid top list regardless of the observation. In a sense, this reconciles both points of view, as an invalid prediction could not outperform any arbitrary valid prediction, thereby disqualifying it in comparison. In what follows, I focus on the construction of consistent scoring functions for valid top lists at first and propose a way of extending such scoring functions to invalid top lists that is less daunting than simply assigning an infinite score.

\section{Mathematical Preliminaries}
\label{Sec:MathPrelim}

This section introduces some preliminary results, which are used heavily in the next section.

\subsection{Symmetric scoring rules}
\label{Sec:SymScoringRules}

The proposed scoring functions are based on symmetric proper scoring rules.
Recall from \cite{GR07} 
that (subject to mild regularity conditions) any proper scoring rule $\myS \colon \cPY \rightarrow \overline{\R}$ admits a \emph{Savage representation},
\begin{equation} \label{Eq:Savage}
\myS(p,y) = G(p) - \langle G'(p), p\rangle + G'_y(p),
\end{equation}
in terms of a concave function $G\colon\Delta_{m-1} \rightarrow \R$ and a supergradient $G'\colon\Delta_{m-1} \rightarrow \R^m$ of $G$, i.e., a function satisfying the \emph{supergradient inequality}
\begin{equation} \label{Eq:Supergrad}
G(q) \leq G(p) + \langle G'(p), q - p\rangle
\end{equation}
for all $p,q \in \Delta_{m-1}$. Conversely, any function of the form \eqref{Eq:Savage} is a proper scoring rule. The function $G$ is strictly concave if, and only if, $\myS$ is strictly proper. It is called the \emph{entropy (function)} of $\myS$, and it is simply the expected score $G(p) = \myE[\myS(p,Y)]$ under the posited distribution, $Y \sim p$. The supergradient inequality \eqref{Eq:Supergrad} is strict if $G$ is strictly concave and $p \neq q$ \citep[][Satz 5.1.12]{Jun15}.

Let $\Sym(\cY)$ denote the symmetric group on $\cY$, i.e., the set of all permutations of $\cY$.
A scoring rule is called \emph{symmetric} if scores are invariant under permutation of classes, i.e.,
\[
\myS((p_y),y) = \myS((p_{\tau^{-1}(y)}),\tau(y))
\]
holds for any permutation $\tau \in \Sym(\cY)$ and all $y \in \cY, p \in \cPY$. Clearly, the entropy function $G$ of a symmetric scoring rule is also symmetric, i.e., invariant to permutation in the sense that
$G(p) = G((p_{\tau(y)}))$
holds for any permutation $\tau \in \Sym(\cY)$ and any distribution $p \in \cPY$.
Vice versa, any symmetric entropy function admits a symmetric proper scoring rule.

\begin{Prop}
	Let $G\colon \cPY \rightarrow\cPY$ be a concave symmetric function. Then there exists a supergradient $G'$ such that the Savage representation \eqref{Eq:Savage} yields a symmetric proper scoring rule.
	\label{Prop:SymSupergrad}
\end{Prop}

\begin{proof}
	Let $\bar{G}'$ be a supergradient of $G$. Using the shorthand $v_\tau = (v_{\tau^{-1}(y)})_{y \in \cY}$ for vectors $v = (v_y)_{y\in\cY} \in \R^m$ indexed by $\cY$ and permutations $\tau \in \Sym(\cY)$, define $G'$ by
	\[
	G'(p) = \frac{1}{\vert \Sym(\cY)\vert} \sum_{\tau \in \Sym(\cY)} \bar{G}'_{\tau^{-1}}(p_\tau)
	\]
	for $p \in \cPY$.
	By symmetry of $G$ and the supergradient inequality,
	\[
	G(q) = G(q_\tau) \leq G(p_\tau) + \langle \bar{G}'(p_\tau),q_\tau - p_\tau\rangle = G(p) + \langle \bar{G}'_{\tau^{-1}}(p_\tau),q - p\rangle
	\]
	holds for all $p,q \in \cPY$ and $\tau \in \Sym(\cY)$. Summation over all $\tau \in \Sym(\cY)$ and division by the cardinality of the symmetric group $\Sym(\cY)$  yields 
	\[
	G(q) \leq \frac{1}{\vert \Sym(\cY)\vert} \sum_{\tau \in \Sym(\cY)} (G(p) + \langle \bar{G}'_{\tau^{-1}}(p_\tau),q - p\rangle) = G(p) + \langle G'(p),q - p\rangle
	\]
	for any $p,q \in \cPY$. Therefore, $G'$ is a supergradient and the Savage representation \eqref{Eq:Savage} yields a symmetric scoring rule, since
	\begin{align*}
	G'(p) &= \frac{1}{\vert \Sym(\cY)\vert} \sum_{\tau \in \Sym(\cY)} \bar{G}'_{\tau^{-1}}(p_\tau) = \frac{1}{\vert \Sym(\cY)\vert} \sum_{\tau \in \Sym(\cY)} \bar{G}'_{(\tau \circ \rho)^{-1}}(p_{\tau \circ \rho}) \\
	&= \frac{1}{\vert \Sym(\cY)\vert} \sum_{\tau \in \Sym(\cY)} \bar{G}'_{\rho^{-1} \circ \tau^{-1}}(p_{\tau \circ \rho}) = \frac{1}{\vert \Sym(\cY)\vert} \sum_{\tau \in \Sym(\cY)} (\bar{G}'_{\tau^{-1}}(p_{\tau \circ \rho}))_{\rho^{-1}} \\
	&= \left(\frac{1}{\vert \Sym(\cY)\vert} \sum_{\tau \in \Sym(\cY)} \bar{G}'_{\tau^{-1}}((p_{\rho})_\tau)\right)_{\rho^{-1}} = G'_{\rho^{-1}}(p_\rho)
	\end{align*}
	and
	\[
	\langle G'(p),p\rangle = \langle G'_{\rho^{-1}}(p_\rho), p \rangle = \langle G'(p_\rho),p_\rho \rangle
	\]
	holds for any permutation $\rho \in \Sym(\cY)$ and all $p \in \cPY$.
\end{proof}

On the other hand, not all proper scoring rules with symmetric entropy function are symmetric. 
The following result provides a necessary condition satisfied by supergradients of symmetric proper scoring rules.

\begin{Lem}
	Let\/ $\myS$ be a symmetric proper scoring rule. If $p \in \Deltam$ satisfies $p_y = p_z$ for $y,z \in \cY$, then the supergradient $G'(p)$ at $p$ in the Savage representation \eqref{Eq:Savage} satisfies $G'_y(p) = G'_z(p)$.
	\label{Lem:SymSuperGrad}
\end{Lem}

\begin{proof}
	Let $\tau = (y~z)$ be the permutation swapping $y$ and $z$ while keeping all other classes fixed.
	Using notation as in the proof of Proposition \ref{Prop:SymSupergrad}, the equality $\myS(p,y) = \myS(p_\tau,\tau(y))$ holds by symmetry of $\myS$. Since $p = p_\tau$, the Savage representation \eqref{Eq:Savage} yields $G'_y(p) = G'_{\tau(y)}(p) = G'_z(p)$.
\end{proof}

The Brier score \eqref{Eq:Brier} and the logarithmic score \eqref{Eq:LogS} are both symmetric scoring rules. The entropy function of the Brier score is given by
\begin{equation}\label{Eq:BrierEntropy}
G(p) = 1 - \sum_{y \in \cY} p_y^2,
\end{equation}
whereas the entropy of the logarithmic score is given by
\[
G(p) = -\sum_{y \in \cY} p_y \log(p_y)
\]
\citep[see][]{GR07}.
\subsection{Majorization and Schur-concavity}

In this section, I adopt some definitions and results on majorization and Schur-concavity from \cite{MOA11}. The theory of majorization is essentially a theory of inequalities, which covers many classical results and a plethora of mathematical applications not only in stochastics.

For a vector $v \in \R^m$, the vector $v_{[\,]} := (v_{[i]})_{i = 1}^m$, where
\[
v_{[1]} \geq \dots \geq v_{[m]}
\]
denote the components of $v$ in decreasing order, is called the \emph{decreasing rearrangement} of $v$. A vector $w \in \R^m$ is a permutation of $v \in \R^m$ (i.e., $w$ is obtained by permuting the entries of $v$) precisely if $v_{[\,]} = w_{[\,]}$.
For vectors $v,w\in\R^m$ with equal sum of components, $\sum_i v_i = \sum_i w_i$, the vector $v$ is said to \emph{majorize} $w$, or $v \succ w$ for short, if the inequality
\[
\sum_{i = 1}^k v_{[i]} \geq \sum_{i = 1}^k w_{[i]}
\]
holds for all $k = 1,\dots,m-1$. 

Let $D \subseteq \R^m$. A function $f \colon D \rightarrow \R$ is \emph{Schur-concave on} $D$ if $v \succ w$ implies $f(v) \leq f(w)$ for all $v,w \in D$. 
A Schur-concave function $f$ is \emph{strictly Schur-concave} if $f(v) < f(w)$ holds whenever $v \succ w$ and $v_{[\,]} \neq w_{[\,]}$.
In particular, any symmetric concave function is Schur-concave 
and strictly Schur-concave if it is strictly concave \citep[][Chapter 3, Proposition C.2 and C.2.c]{MOA11}. Hence, the following lemma holds.

\begin{Lem}
	The entropy function of any symmetric proper scoring rule is Schur-concave. It is strictly Schur-concave if the scoring rule is strictly proper.
	\label{Lem:EntropySchur}
\end{Lem}

A set $D \subset \R^m$ is called \emph{symmetric} if $v \in D$ implies $w \in D$ for all vectors $w \in \R^m$ such that $v_{[\,]} = w_{[\,]}$.
By the Schur-Ostrowski criterion \citep[][Chapter 3, Theorem A.4 and A.4.a]{MOA11} a continuously differentiable function $f\colon D \rightarrow \R$ on a symmetric convex set $D$ with non-empty interior is Schur-concave if, and only if, $f$ is symmetric and the partial derivatives $f_{(i)}(v) = \frac{\partial}{\partial v_i} f(v)$ increase as the components $v_i$ of $v$ decrease, i.e., $f_{(i)}(v) \leq f_{(j)}(v)$ if (and only if) $v_i \geq v_j$.

Unfortunately, this does not hold for supergradients of concave functions. The following is a slightly weaker condition, which applies to supergradients of symmetric concave functions.

\begin{Lem}[Schur-Ostrowski condition for concave functions]
	Let $f\colon D \rightarrow \R$ be a symmetric concave function on a symmetric convex set $D$, $v \in D$ and $f'(v)$ be a supergradient of $f$ at $v$, i.e., a vector satisfying the supergradient inequality
	\begin{equation}
	f(w) \leq f(v) + \langle f'(v), w - v\rangle
	\label{Eq:GeneralSupergrad}
	\end{equation}
	for all $w \in D$.
	Then $v_i > v_j$ implies $f'_i(v) \leq f'_j(v)$.
	\label{Lem:SchurOstrowski}
\end{Lem}

\begin{proof}
	For $i = 1,\dots,m$, let $e_i = (\one\{i = j\})_{j = 1}^m$ denote the $i$-th vector of the standard basis of $\R^m$.
	Let $v \in D$ be such that $v_i > v_j$ for some indices $i,j$ and let $0 < \varepsilon \leq v_i - v_j$. Define $w = v - \varepsilon e_i + \varepsilon e_j$. Then $v \succ w$ \citep[by][Chapter 2, Theorem B.6]{MOA11}, because $w$ is obtained from $v$ through a so called `$T$-transformation' \citep[see][p.\ 32]{MOA11}, i.e., $w_i = \lambda v_i + (1-\lambda) v_j$ and $w_j = \lambda v_j + (1-\lambda) v_i$ with $\lambda = \frac{v_i - v_j - \varepsilon}{v_i - v_j}$. By Schur-concavity of $f$, this implies $f(v) \leq f(w)$ and the supergradient inequality \eqref{Eq:GeneralSupergrad} yields
	\[
	\varepsilon (f'_j(v) - f'_i(v)) = \langle f'(v),w - v \rangle \geq f(w) - f(v) \geq 0.
	\]
	Hence, the inequality $f'_j(v) \geq f'_i(v)$ holds.
\end{proof}

With this, there is no need to restrict attention to differentiable entropy functions when applying the Schur-Ostrowski condition in what follows.
Furthermore, true top-$k$ lists can be characterized using majorization.

\begin{Lem}
	Let $Y \sim p$ be distributed according to $p \in \cPY$. The padded distribution $\tilde{t}$ associated with a true top-$k$ list $t \in \myT_k(p)$ majorizes the padded distribution $\tilde{s}$ associated with any calibrated top-$k$ list $s \in \cT_k$.
	\label{Lem:TrueMajor}
\end{Lem}

\begin{proof}
	The sum of confidence scores $\sum_{i = 1}^k \tilde{t}_{[i]} = \sum_{i = 1}^k p_{[i]} \geq \sum_{i = 1}^k \tilde{s}_{[i]}$ of a true top-$k$ list is maximal among calibrated top-$k$ lists by definition. Hence, the confidence score $\hat{t}_{[i]} = \tilde{t}_{[i]}$ of the true top-$k$ list $t = (\hat{Y},\hat{t})$ matches the $i$-th largest class probability $p_{[i]}$ for $i = 1,\dots,k$. Therefore, the partial sums $\sum_{i = 1}^{\ell} \tilde{t}_{[i]} = \sum_{i = 1}^{\ell} p_{[i]} \geq \sum_{i = 1}^{\ell} \tilde{s}_{[i]}$ across the largest confidence scores are also maximal for $\ell = 1,\dots,k-1$. Furthermore, the proxy probability $\pi(t) = \frac{1 - \sum_{i = 1}^k \tilde{t}_{[i]}}{m-k}$ associated with a true top-$k$ list is minimal among calibrated top-$k$ lists. Hence, the partial sums 
	\[
	\sum_{i = 1}^{\ell} \tilde{t}_{[i]} = 1 - (m-\ell)\pi(t) \geq 1 - (m-\ell)\pi(s) = \sum_{i = 1}^{\ell} \tilde{s}_{[i]}
	\]
	are maximal for $\ell > k$ .
\end{proof}

\section{Consistent Top List Scores}
\label{Sec:PaddedScores}

Having reviewed the necessary preliminaries, this section shows that the proposed padded symmetric scores constitute a family of consistent scoring functions for the probabilistic top list functionals. The padded symmetric scores are defined for valid top lists and can be extended to invalid top lists by scoring the largest valid sublist, which yields a consistent scoring function. Strict consistency is preserved by adding an additional penalty term to the score of an invalid prediction.

\subsection{Padded symmetric scores}

From now on, let $\myS\colon\cPY \rightarrow \overline{\R}$ be a proper symmetric scoring rule with entropy function $G$.
The scoring rule $\myS$ is extended to valid top-$k$ lists for $k = 0,1,\dots,m-1$ by setting
\[
\myS(t,y) := \myS(\tilde{t},y)
\]
for $y \in \cY,t \in \tilde{\cT}_k$, where $\tilde{t} \in \Delta_{m-1}$ is the padded distribution \eqref{Eq:PaddedDist} associated with the top-$k$ list $t$.
I call the resulting score $\myS\colon \bigcup_{k = 0}^m \tilde{\cT}_k \times \cY \rightarrow \overline\R$ a \emph{padded symmetric score}.
For example, the logarithmic score \eqref{Eq:LogS} yields the \emph{padded logarithmic score}
\begin{equation*} 
\myS_{\log}((\hat{Y},\hat{t}),y) 
=  \begin{cases}
- \log(\hat{t}_y),& \text{if } y \in \hat{Y}, \\
\log(m-k) - \log(1-\sum_{z \in \hat{Y}} \hat{t}_z),& \text{otherwise,}
\end{cases}
\end{equation*}
whereas the Brier score \eqref{Eq:Brier} yields the \emph{padded Brier score}
\begin{equation}\label{Eq:PaddedBrier}
\myS_{\mathrm{B}}((\hat{Y},\hat{t}),y) 
= 1 + \sum_{z\in\hat{Y}} \hat{t}_z^2 + \frac{(1-\sum_{z \in \hat{Y}} \hat{t}_z)^2}{m-k} - 2 \cdot
\begin{cases}
\hat{t}_y,				&\text{if } y \in \hat{Y}, \\
\frac{1-\sum_{z \in \hat{Y}} \hat{t}_z}{m-k},	&\text{otherwise.}
\end{cases}
\end{equation}
The following example shows that padded symmetric scores should not be applied to invalid top lists without further considerations.

\begin{Ex}
	If a padded symmetric score based on a strictly proper scoring rule is used to evaluate the invalid top-2 list $s$ in Example \ref{Ex:TopLists}, it attains a lower expected score than a true top list $t \in \myT_2(p)$, because $\tilde s = p$, whereas $\tilde{t} \neq p$. Hence, the score would fail to be consistent.
\end{Ex}

The following lemma shows that the expected score of a calibrated top list is fully determined by the top list itself and does not depend on (further aspects of) the underlying distribution.

\begin{Lem}
	Let $\myS$ be a padded symmetric score. If $p \in \cPY$ is the true distribution of\/ $Y \sim p$, and\/ $t$ is a calibrated valid top list, then the expected score of the top list\/ $t$ matches the entropy of the padded distribution\/ $\tilde{t}$,
	\[
	\E[\myS(t,Y)] = G(\tilde{t}).
	\]
	\label{Lem:ExpEntropy}
\end{Lem}

\begin{proof}
	Let $t = (\hat{Y},\hat{t}) \in \tilde{\cT}_k(p)$. Assume w.l.o.g.\ $k < m$ (the claim is trivial if $k = m$) and let $z \in \cY\setminus\hat{Y}$. By Lemma \ref{Lem:SymSuperGrad} the supergradient at $\tilde{t}$ satisfies $G'_y(\tilde{t}) = G'_z(\tilde{t})$ for all $y \notin \hat{Y}$. Hence, the Savage representation \eqref{Eq:Savage} of the underlying scoring rule yields
	\begin{align*}
	\myE[\myS(t,Y)] &=  G(\tilde{t}) - \langle G'(\tilde{t}), \tilde{t}\rangle 
	+ \sum_{y \in \cY} p_y G'_y(\tilde{t}) \\
	&= G(\tilde{t}) - \sum_{y \in \hat{Y}} (p_y - \hat{t}_y) G'_y(\tilde{t}_y) - \left(\sum_{y \notin \hat{Y}} p_y - (m-k)\pi(t)\right) G'_z(\tilde{t}) = G(\tilde{t}),
	\end{align*}
	because $t$ is calibrated.
\end{proof}

Padded symmetric scores exhibit an interesting property that admits balanced comparison of top list predictions of varying length. A top list score $\myS\colon \bigcup_{k = 0}^m \tilde{\cT}_k \times \cY \rightarrow \overline\R$ exhibits the \emph{comparability property} 
if the expected score does not deteriorate upon extending a true top list, i.e., for $k = 0,1,\dots,m-1$ and any distribution $p\in\cPY$ of $Y \sim p$,
\begin{equation}
\myE[\myS(t_{k+1},Y)] \leq \myE[\myS(t_k,Y)]
\label{Eq:Comparability}
\end{equation}
holds for $t_k \in \myT_k(p)$ and $t_{k+1}\in \myT_{k+1}(p)$. The following theorem shows that padded symmetric scores in fact exhibit the comparability property. I use the comparability property to show consistency of the individual padded symmetric top-$k$ list scores $\myS\vert_{\tilde{\cT}_k \times \cY}$ and to extend these scores to invalid top lists. Section \ref{Sec:Comp} provides further discussion and some numerical insights.

\begin{theorem}
	Padded symmetric scores exhibit the comparability property.
	\label{Thm:CompSym}
\end{theorem}

\begin{proof}
	Let $\myS$ be a padded symmetric score and $G$ be the concave entropy function of the underlying proper scoring rule. 
	Let $Y \sim p$ be distributed according to some distribution $p \in \cPY$ and let $t_k = (\hat{Y}_k,(p_y)_{y \in \hat{Y}_k})$ be a calibrated valid top-$k$ list for some $k = 0,1,\dots,m-1$, which is extended by a calibrated valid top-$(k+1)$ list $t_{k+1} = (\hat{Y}_{k+1},(p_y)_{y \in \hat{Y}_{k+1}})$ in the sense that $\hat{Y}_{k+1} = \hat{Y}_k \cup \{z\}$ for some $z\in\cY$.
	It is easy to verify that
	$\tilde t_{k+1} \succ \tilde t_k$, since $p_z \geq \pi(t_k) \geq \pi(t_{k+1})$. 
	Hence, the inequality $G(\tilde{t}_{k+1}) \leq G(\tilde{t}_k)$ holds by Schur-concavity of $G$ (Lemma \ref{Lem:EntropySchur}), which yields the desired inequality of expected scores by Lemma \ref{Lem:ExpEntropy}.
	
	Clearly, there exists a true top-$(k+1)$ list $t_{k+1}\in \myT_{k+1}(p)$ extending a true top-$k$ list $t_k \in \myT_k(p)$ in the above sense. By a symmetry argument all true top lists of a given length have the same expected score and hence $\myS$ exhibits the comparability property. 
\end{proof}

Note that the proof of Theorem \ref{Thm:CompSym} shows that \eqref{Eq:Comparability} holds for any calibrated valid extension $t_{k+1}$ of a calibrated valid top list $t_k$ and not only true top lists. I proceed to show that padded symmetric scores restricted to valid top-$k$ lists are consistent for the top-$k$ list functional.

\begin{Thm}
	Let $k \in \{0,1,\dots,m\}$ be fixed and $\myS\colon \bigcup_{\ell = 0}^m \tilde{\cT}_\ell \times \cY \rightarrow \overline\R$ be a padded symmetric score. 
	Then the restriction $\myS\vert_{\tilde{\cT}_k \times \cY}$ of the score $\myS$ to the set of valid top-$k$ lists $\tilde{\cT}_k$ is consistent for the top-$k$ list functional $\myT_k$. It is strictly consistent if the underlying scoring rule $\myS\vert_{\cPY \times \cY}$ is strictly proper.
	\label{Thm:PaddedSym}
\end{Thm}

\begin{proof}
	Let $p = (p_y)_{y \in \cY} \in \cPY$ be the true probability distribution of $Y \sim p$.
	Clearly, all true top-$k$ lists in $\myT_k(p)$ attain the same expected score by symmetry of the underlying scoring rule. Let $t = (\hat{Y},(p_y)_{y \in \hat{Y}}) \in \myT_k(\hat{p})$ be a true top-$k$ list and $s = (\hat{Z},(\hat{s}_y)_{y \in \hat{Z}}) \in \tilde{\cT}_k$ be an arbitrary valid top-$k$ list. To show consistency of $\myS\vert_{\tilde{\cT}_k \times\cY}$, it suffices to show that the valid top-$k$ list $s$ does not attain a lower (i.e., better) expected score than the true top-$k$ list $t$. Strict consistency follows if the expected score of any $s \notin \myT_k(p)$ is higher than that of the true top-$k$ list $t$.
	
	First, consider $s \notin \myT_k(p)$ to be a calibrated top-$k$ list, i.e., $\hat{s}_y = p_y$ for all $y \in \hat{Z}$. Since $\tilde t$ majorizes $\tilde s$ by Lemma \ref{Lem:TrueMajor}, the inequality 
	\[
	\E [\myS(t,Y)] = G(\tilde t) \leq G(\tilde s) = \E [\myS(s,Y)]
	\]
	holds by Schur-concavity of the entropy function $G$ (Lemma \ref{Lem:EntropySchur}) and Lemma \ref{Lem:ExpEntropy}.
	If the underlying scoring rule is strictly proper, the entropy function is strictly (Schur-)concave, and hence the inequality is strict.
	
	Now, consider $s$ to be an uncalibrated top-$k$ list and let $r = (\hat{Z},(p_y)_{y \in \hat{Z}
	})$ be the respective calibrated top-$k$ list on the same classes.
	The calibrated top-$k$ list $r$ may not be valid and cannot be scored if this is the case. However, its largest valid sublist $r' = (\hat{Z}', (p_y)_{y \in \hat{Z}'})$ with $\hat{Z}' \subseteq \hat{Z}$ can be scored. 
	Let $z \in \cY\setminus \hat{Z}$. The difference in expected scores
	\begin{align}
	&\E [\myS(s, Y)] - \E [\myS(r',Y)] \notag \\
	&= G(\tilde s) - G(\tilde r') - \langle G'(\tilde s), \tilde s \rangle + \langle G'(\tilde  r'),  \tilde r' \rangle + \sum_{y \in \cY} p_y (G_y'(\tilde s) - G_y'(\tilde  r')) \tag{by the Savage representation \eqref{Eq:Savage}} \\ 
	&\geq \langle G'(\tilde  r') - G'(\tilde s),  \tilde r' \rangle + \sum_{y \in \cY} p_y (G_y'(\tilde s) - G_y'(\tilde r')) \tag{by the supergradient inequality \eqref{Eq:Supergrad}} \\ 
	&= \sum_{y \in \hat{Z} \setminus \hat{Z}'} (p_y - \pi(r') ) (G_y'(\tilde s) - G_z'(\tilde  r')) + \sum_{y \in \cY \setminus \hat{Z}} (p_y - \pi(r') ) (G_z'(\tilde s) - G_z'(\tilde  r')) \tag{by Lemma \ref{Lem:SymSuperGrad}} \\ 
	&= \sum_{y \in \hat{Z} \setminus \hat{Z}'} (p_y - \pi(r') ) (G_y'(\tilde s) - G_z'(\tilde s)) \tag{as $\sum_{y \in \cY \setminus \hat{Z}} (p_y - \pi(r')) = - \sum_{y \in \hat{Z} \setminus \hat{Z}'} (p_y - \pi(r'))$} 
	\end{align}
	is nonnegative by the fact that $(p_y - \pi(r') ) \leq 0$ for $y \in \hat{Z}\setminus\hat{Z}'$ (since $r'$ is the largest valid sublist) and Lemma \ref{Lem:SchurOstrowski} (and Lemma \ref{Lem:SymSuperGrad} if $\hat{s}_y = \pi(s)$ for some $y \in \hat{Z}\setminus\hat{Z}'$).
	
	Let $k' = \vert \hat{Z}'\vert$. Then, $r'$ scores no better than a true top-$k'$ list $t_{k'} \in \myT_{k'}(p)$, which in turn scores no better than $t$ by the comparability property. Therefore,
	\[
	\E [\myS(s, Y)] \geq \E [\myS(r',Y)] \geq \E [\myS(t_{k'}, Y)] \geq \E [\myS(t,Y)]
	\]
	holds. 	
	If the underlying scoring function is strictly proper, the difference in expected scores $\E [\myS(s, Y)] - \E [\myS(r',Y)]$ above is strictly positive by strictness of the supergradient inequality \citep[][Satz 5.1.12]{Jun15}, and hence $\E [\myS(t,Y)] < \E [\myS(s, Y)]$ holds in this case, which concludes the proof.
\end{proof}

\subsection{Penalized extensions of padded symmetric scores}

The comparability property can be used to extend a padded symmetric score $\myS$ to invalid top lists in a consistent manner. To this end, recall that $t'$ denotes the largest valid sublist of a top list $t = (\hat{Y},\hat{t}) \in \cT_k$. Assigning the score of the largest valid sublist to an invalid top-$k$ list yields a consistent score by the comparability property.  Strict consistency of the padded symmetric score $\myS$ is preserved by adding a positive penalty term $c_{\mathrm{invalid}} > 0$ to the score of the largest valid sublist in the case of an invalid top list prediction. I call the resulting score extension $\myS\colon \bigcup_{k = 0}^m \cT_k \times \cY \rightarrow \overline\R$, which assigns the score
\begin{equation} \label{Eq:PenExt}
\myS(t,y) = \myS(t',y) + c_{\mathrm{invalid}}
\end{equation}
to an invalid top list $t \in \cT_k \setminus \tilde{\cT}_k$ for $k = 1,2,\dots,m-1$, a \emph{penalized extension} of a padded symmetric score.
The following example illustrates that the positive penalty is necessary to obtain a strictly consistent scoring function.

\begin{Ex}
	Consider a setting similar to that of Example \ref{Ex:TopLists} with $Y \sim p = (0.4,0.2,0.2,0.2)$. The padded distribution associated with the largest valid sublist $t' = (\{1\},(0.4))$ of the invalid list $t = (\{1,2\},(0.4,0.1))$ matches the true distribution, $\tilde{t}' = p$, and hence the expected score of $t$ in \eqref{Eq:PenExt} is minimal if $c_{\mathrm{invalid}} = 0$.
\end{Ex}

The following theorem summarizes the properties of the proposed score extension.

\begin{theorem} Let $k \in \{0,1,\dots,m\}$ be fixed and $\myS\colon\bigcup_{\ell = 0}^m \cT_\ell \times \cY \rightarrow \overline\R$ be a penalized extension \eqref{Eq:PenExt} of a padded symmetric score with penalty term $c_{\mathrm{invalid}} \geq 0$.
	Then the restriction $\myS\vert_{\cT_k \times \cY}$ of the score $\myS$ to the set of top-$k$ lists $\cT_k$ is consistent for the top-$k$ list functional $\myT_k$. It is strictly consistent if the underlying scoring rule $\myS\vert_{\cPY \times \cY}$ is strictly proper and the penalty term $c_{\mathrm{invalid}}$ is nonzero.
	\label{Thm:PenSym}
\end{theorem}

\begin{proof}
	In light of Theorem \ref{Thm:PaddedSym}, it remains to show that an invalid top-$k$ list attains a worse expected score than a true top-$k$ list $t \in \myT_k(p)$ under the true distribution $p \in \cPY$ of $Y \sim p$. To this end, let $s \in \cT_k$ be invalid. By construction of the penalized extension, the top list $s$ is assigned the score of its largest valid sublist $s'$ plus the additional penalty $c_{\mathrm{invalid}}$. By consistency of the padded symmetric score and the comparability property, the expected score of $s'$ cannot fall short of the expected score of $t$. Hence, $\myS\vert_{\cT_k \times \cY}$ is consistent for the top-$k$ list functional. If a positive penalty $c_{\mathrm{invalid}} > 0$ is added, the score extension is strictly consistent given a strictly consistent padded symmetric score.
\end{proof}

\section{Comparability}
\label{Sec:Comp}

The comparability property \eqref{Eq:Comparability} ensures that additional information provided by an extended true top list does not adversely influence the expected score. The information gain is quantified by a reduction in entropy, which depends on the underlying scoring rule. Ideally, a top list score encourages the prediction of classes that account for a substantial portion of probability mass, while offering little incentive to provide unreasonably large top lists. In what follows, I argue that the padded Brier score satisfies this requirement.

Let $\myS$ be a padded symmetric score with entropy function $G$ (of the underlying proper scoring rule). Furthermore, let $1\leq k<m$ and $t = (\hat{Y},(\hat{t}_y)_{y\in\hat{Y}})$ be a top-$k$ list that accounts for most of the probability mass. In particular, assume that the unaccounted probability $\alpha = \alpha(t) = 1 - \sum_{y \in \hat{Y}} \hat{t}_y$ is less than the least confidence score but nonzero, i.e.,
\begin{equation}\label{Eq:AssAlpha}
0 < \alpha < \min_{y \in \hat{Y}} \hat{t}_y. 
\end{equation}
Let $Q = Q(t) = \{p \in \cPY \mid t \in \myT_k(p)\}$ be the set of all probability measures relative to which $t$ is a true top-$k$ list. Let $p \in Q$ assign the remaining probability mass $\alpha$ to a single class. Then $p$ majorizes any $q \in Q$, and the distribution $p$ has the lowest entropy, i.e., $G(p) = \min_{q\in Q} G(q)$, by Schur-concavity of the entropy function (Lemma \ref{Lem:EntropySchur}). As the expected score of the top list $t$ is invariant under distributions in $Q$ by Lemma \ref{Lem:ExpEntropy}, the relative difference in expected scores between the true top list $t$ and the true distribution $q \in Q$ is bounded by the relative difference in expected scores between $t$ and $p$,
\[
\frac{G(\tilde{t}) - G(q)}{G(q)} \leq \frac{G(\tilde{t}) - G(p)}{G(p)}.
\]
The upper bound can be simplified by bounding the entropy of $p$ from below, as $G(p) \geq G((1-\alpha,\alpha,0,\dots,0))$ by Schur-concavity of $G$.

If $\myS = \myS_{\mathrm{B}}$ is the padded Brier score \eqref{Eq:PaddedBrier} with entropy \eqref{Eq:BrierEntropy}, the lower bound reduces to $G(p) \geq G((1-\alpha,\alpha,0,\dots,0)) = 2(\alpha - \alpha^2) > \alpha$, since $\alpha < 0.5$ by assumption \eqref{Eq:AssAlpha} and hence $2\alpha^2 < \alpha$. With this the relative difference in expected scores has a simple upper bound,
\[
\frac{G(\tilde{t}) - G(p)}{G(p)} = \frac{\alpha^2 - \alpha \pi(t)}{2(\alpha - \alpha^2)} < \frac{\alpha^2}{\alpha} = \alpha.
\]
For the padded logarithmic score no such bound exists and the deviation of the expected top list score from the optimal score can be severe, as illustrated in the following numerical example.
The example sheds some light on the behavior of the (expected) padded symmetric scores and demonstrates that top lists of length $k > 1$ may provide valuable additional information over a simple mode prediction.

\begin{example} \label{Ex:PaddedScores}
	\begin{table}[t]
		\caption{Expected padded Brier scores and expected padded logarithmic scores of various types of true predictions and multiple distributions discussed in Example \ref{Ex:PaddedScores}. Relative score differences (in percent) with respect to the optimal scores are in brackets.}
		\label{Tab:ExScores}
		\vspace{2mm}
		\centering 
		\begin{tabular}{cccccc}
			\toprule
			&& \multicolumn{4}{c}{$\myE[\myS(\cdot,Y)]$} \\
			$p$ & $\myS$ & $\operatorname{Mode}(p)$ & $\myT_1(p)$ & $\myT_2(p)$ & $p$ \\
			\midrule
			$p^{(\mathrm h)}$ & $\myS_\mathrm{B}$ & 0.02 (1.01\%) & 0.0199 (0.38\%) & 0.0198 (0\%) & 0.0198 \\ 
			$p^{(\mathrm m)}$ & $\myS_\mathrm{B}$ & 1 (70.59\%) & 0.6875 (17.28\%) & 0.5867 (0.08\%) & 0.5862 \\ 
			$p^{(\mathrm l)}$ & $\myS_\mathrm{B}$ & 1.5 (88.87\%) & 0.7969 (0.34\%) & 0.7955 (0.16\%) & 0.7942 \\ 
			\midrule
			$p^{(\mathrm h)}$ & $\myS_{\log}$ & $\infty$ & 0.0699 (24.75\%) & 0.0560 (0\%) & 0.0560 \\
			$p^{(\mathrm m)}$ & $\myS_{\log}$ & $\infty$ & 1.3863 (32.49\%) & 1.0532 (0.66\%) & 1.0463 \\
			$p^{(\mathrm l)}$ & $\myS_{\log}$ & $\infty$ & 1.6021 (0.45\%) & 1.5984 (0.23\%) & 1.5948 \\
			\bottomrule
		\end{tabular}
	\end{table}
	
	Suppose there are $m = 5$ classes labeled $1,2,\dots, 5$ and the true (conditional) distribution $p = p(\mathbf{x}) = \cL(Y\mid \bX = \mathbf{x})$ of $Y$ (given a feature vector $\mathbf{x} \in \cX$) is known. Table \ref{Tab:ExScores} features expected padded Brier and logarithmic scores of various types of truthful predictions under several distributions, as well as relative differences with respect to the optimal score.
	The considered distributions
	\begin{equation*}
	\begin{gathered}
	p^{(\mathrm{h})} = (0.99,0.01,0,0,0),
	\quad p^{(\mathrm{m})} = (0.5,0.44,0.03,0.02,0.01), 
	\\
	p^{(\mathrm{l})} = (0.25,0.22,0.2,0.18,0.15).
	\end{gathered}
	\end{equation*}
	exhibit varying degrees of predictability.
	Distribution $p^{(\mathrm h)}$ exhibits high predictability in the sense that a single class can be predicted with high confidence. Distribution $p^{(\mathrm m)}$ exhibits moderate predictability in that it is possible to narrow predictions down to a small subset of classes with high confidence, but getting the class exactly right is a matter of luck. Distribution $p^{(\mathrm l)}$ exhibits low predictability in the sense that all classes may well realize.
	Predictions are of increasing information content. The first prediction is the true mode, i.e., a hard classifier without uncertainty quantification that predicts class 1 under all considered distributions. The hard mode is interpreted as assigning all probability mass to the predicted class. Scores are obtained by embedding the predicted class in the probability simplex or, equivalent, by scoring the top-1 list $(\{1\},1)$. The second prediction is the true top-1 list $(\{1\},p_1)$, i.e., the mode with uncertainty quantification. The third prediction is the true top-2 list $(\{1,2\},(p_1,p_2))$ and the final prediction is the true distribution $p$ itself.
	
	By consistency of the padded symmetric scores, the true top-1 lists score better in expectation than the mode predictions and by the comparability property, the true top-2 lists score better than the top-1 lists, while the true distributions attain the optimal scores.
	The mode predictions perform significantly worse than the probabilistic predictions, which highlights the importance of truthful uncertainty quantification. Note that the log score assigns an infinite score in case of the true outcome being predicted as having zero probability, hence the mode prediction is assigned an infinite score with positive probability. 
	
	The expected padded Brier score of the probabilistic top-1 list under the highly predictable distribution $p^{(\mathrm{h})}$ is not far from optimal, whereas the respective logarithmic score is inflated by the discrepancies between the padded and true distributions, even though the top list accounts for most of the probability mass ($\alpha = 0.01$). Deviations from the optimal scores are more pronounced under the logarithmic score in all considered cases.
	
	Under the distribution exhibiting moderate predictability, the top-2 list prediction is much more informative than the top-1 list prediction, which results in a significantly improved score that is not far from optimal. 
	Under the distribution exhibiting low predictability, all probabilistic predictions perform well, as there is little information to be gained.
\end{example}

Estimation of small probabilities is frequently hindered by finite sample size. The specification of top list predictions in conjunction with the padded Brier score circumvents this issue, as the Brier score is driven by absolute differences in probabilities, whereas the logarithmic score emphasizes relative differences in probabilities. In other words, the padded distribution is deemed a good approximation of the true distribution if the true top list accounts for most of the probability mass by the Brier score.

In light of these considerations, I conclude that the padded Brier score is suitable for the comparison of top list predictions of varying length.

\section{Concluding Remarks}
\label{Sec:ConclusionTopLists}

In this paper, I argued for the use of evaluation metrics rewarding truthful probabilistic assessments in classification. To this end, I introduced the probabilistic top list functionals, which offer a flexible probabilistic framework for the general classification problem. Padded symmetric scores yield consistent scoring functions, which admit comparison of various types of predictions. The padded Brier score appears particularly suitable, as top lists accounting for most of the probability mass obtain an expected padded Brier score that is close to optimal. 

The entropy of a distribution is a measure of uncertainty or information content. Majorization provides a relation characterizing common decreases in entropy shared by all symmetric proper scoring rules.  In particular, for two distributions $p \in \cPY$ and $q \in \cPY$, the entropy of the distribution $p$ does not exceed the entropy of $q$, i.e., $G(p) \leq G(q)$, if $p$ majorizes $q$. The inequality is strict if the scoring rule is strictly proper and $q$ is not a permutation of $p$.

Similar to probabilistic top-$k$ lists, a probabilistic top-$\beta$ list with $\beta \in (0,1)$ may be defined as a minimal top list accounting for a probability mass of at least $\beta$. However, the padded symmetric scores proposed in this paper are not consistent for the top-$\beta$ list functional, and the question whether this functional is elicitable constitutes an open problem for future research.

As a simple alternative to the symmetric padded scores proposed in this paper, top-$k$ error \citep{YK20} is also a consistent scoring function for the top-$k$ list functional, however, it is not strictly consistent, as it does not evaluate the confidence scores. On a related note, strictly proper scoring rules are essentially top-$k$ consistent surrogate losses in the sense of \cite{YK20}. The idea of a consistent surrogate loss is to find a loss function that is easier to optimize than the target accuracy measure such that the confidence scores optimize accuracy. However, confidence scores need not represent probabilities. In contrast, strictly proper scoring rules elicit probabilities. Essentially, strictly proper scoring rules are consistent surrogates for any loss or scoring function that is consistent for a statistical functional.

Typically, classes cannot simply be averaged. Therefore, combining multiple class predictions may be difficult, as majority voting may result in a tie, while learning individual voting weights or a meta-learner requires training data \citep[see][Section 8.3 for a review of classifier combination techniques]{KZP06}. Probabilistic top lists facilitate the combination of multiple predictions, as confidence scores can simply be averaged, which may be an easy way to improve the prediction.

The prediction of probabilistic top lists appears particularly useful in problems, where classification accuracy is not particularly high, as is frequently the case in multi-label classification.
Probabilistic predictions are an informative alternative to classification with reject option. Furthermore, if it is possible to predict top lists of arbitrary length, the empty top-0 list may be seen as a reject option.
Shifting focus towards probabilistic predictions may well increase prediction quality and usefulness in various decision problems, where misclassification losses are not uniform. The padded symmetric scores serve as general purpose evaluation metrics that account for the additional value provided by probabilistic assessments. Applying the proposed scores in a study with real predictions (e.g., the study conducted by \cite{LPAWQ20}) is left as a topic for future work.

\bibliographystyle{abbrvnat}
\bibliography{manuscript}

\begin{thebibliography}{27}
\providecommand{\natexlab}[1]{#1}
\providecommand{\url}[1]{\texttt{#1}}
\expandafter\ifx\csname urlstyle\endcsname\relax
  \providecommand{\doi}[1]{doi: #1}\else
  \providecommand{\doi}{doi: \begingroup \urlstyle{rm}\Url}\fi

\bibitem[Barutcuoglu et~al.(2006)Barutcuoglu, Schapire, and Troyanskaya]{BST06}
Z.~Barutcuoglu, R.~E. Schapire, and O.~G. Troyanskaya.
\newblock {Hierarchical multi-label prediction of gene function}.
\newblock \emph{Bioinformatics}, 22:\penalty0 830--836, 2006.

\bibitem[Brier(1950)]{Brier1950}
G.~W. Brier.
\newblock Verification of forecasts expressed in terms of probability.
\newblock \emph{Monthly Weather Review}, 78:\penalty0 1--3, 1950.

\bibitem[Chen et~al.(2019)Chen, Wei, Wang, and Guo]{CWWG19}
Z.-M. Chen, X.-S. Wei, P.~Wang, and Y.~Guo.
\newblock Multi-label image recognition with graph convolutional networks.
\newblock In \emph{Proceedings of the IEEE/CVF Conference on Computer Vision
  and Pattern Recognition (CVPR)}, 2019.

\bibitem[Dawid(1984)]{Daw84}
A.~P. Dawid.
\newblock Statistical theory: The prequential approach.
\newblock \emph{Journal of the Royal Statistical Society Series A},
  147:\penalty0 278--292, 1984.

\bibitem[Elkan(2001)]{Elkan01}
C.~Elkan.
\newblock The foundations of cost-sensitive learning.
\newblock In \emph{Proceedings of the Seventeenth International Joint
  Conference on Artificial Intelligence}, pages 973--978, 2001.

\bibitem[Gneiting(2011)]{Gneiting2011}
T.~Gneiting.
\newblock Making and evaluating point forecasts.
\newblock \emph{Journal of the American Statistical Association}, 106:\penalty0
  746--762, 2011.

\bibitem[Gneiting(2017)]{G17}
T.~Gneiting.
\newblock When is the mode functional the {B}ayes classifier?
\newblock \emph{Stat}, 6:\penalty0 204--206, 2017.

\bibitem[Gneiting and Katzfuss(2014)]{GK14}
T.~Gneiting and M.~Katzfuss.
\newblock Probabilistic forecasting.
\newblock \emph{Annual Review of Statistics and Its Application}, 1:\penalty0
  125--151, 2014.

\bibitem[Gneiting and Raftery(2007)]{GR07}
T.~Gneiting and A.~E. Raftery.
\newblock Strictly proper scoring rules, prediction, and estimation.
\newblock \emph{Journal of the American Statistical Association}, 102:\penalty0
  359--378, 2007.

\bibitem[Gneiting and Resin(2021)]{GR21}
T.~Gneiting and J.~Resin.
\newblock Regression diagnostics meets forecast evaluation: {C}onditional
  calibration, reliability diagrams and coefficient of determination, 2021.
\newblock Preprint,
  \href{https://arxiv.org/abs/2108.03210v3}{arXiv:2108.03210v3}.

\bibitem[Guo et~al.(2017)Guo, Pleiss, Sun, and Weinberger]{GPSW17}
C.~Guo, G.~Pleiss, Y.~Sun, and K.~Q. Weinberger.
\newblock On calibration of modern neural networks.
\newblock In \emph{Proceedings of the 34th International Conference on Machine
  Learning (ICML)}, 2017.

\bibitem[Herbei and Wegkamp(2006)]{HW06}
R.~Herbei and M.~H. Wegkamp.
\newblock Classification with reject option.
\newblock \emph{The Canadian Journal of Statistics}, 34:\penalty0 709--721,
  2006.

\bibitem[Hui and Belkin(2021)]{HB21}
L.~Hui and M.~Belkin.
\newblock Evaluation of neural architectures trained with square loss vs
  cross-entropy in classification tasks.
\newblock In \emph{International Conference on Learning Representations}, 2021.
\newblock URL \url{https://openreview.net/forum?id=hsFN92eQEla}.

\bibitem[Jungnickel(2015)]{Jun15}
D.~Jungnickel.
\newblock \emph{Optimierungsmethoden: {E}ine {E}inf{\"u}hrung}.
\newblock Springer Spektrum, Berlin, Heidelberg, 2015.

\bibitem[Kotsiantis et~al.(2006)Kotsiantis, Zaharakis, and Pintelas]{KZP06}
S.~B. Kotsiantis, I.~D. Zaharakis, and P.~E. Pintelas.
\newblock Machine learning: a review of classification and combining
  techniques.
\newblock \emph{Artificial Intelligence Review}, 26:\penalty0 159--190, 2006.

\bibitem[Li et~al.(2020)Li, Pavlu, Aslam, Wang, and Qin]{LPAWQ20}
C.~Li, V.~Pavlu, J.~Aslam, B.~Wang, and K.~Qin.
\newblock Learning to calibrate and rerank multi-label predictions.
\newblock In \emph{Machine Learning and Knowledge Discovery in Databases},
  2020.

\bibitem[Marshall et~al.(2011)Marshall, Olkin, and Arnold]{MOA11}
A.~W. Marshall, I.~Olkin, and B.~C. Arnold.
\newblock \emph{Inequalities: {T}heory of Majorization and Its Applications}.
\newblock Springer Series in Statistics. Springer, New York, second edition,
  2011.

\bibitem[Murphy(1977)]{Murphy1977}
A.~H. Murphy.
\newblock The value of climatological, categorical and probabilistic forecasts
  in the cost-loss ratio situation.
\newblock \emph{Monthly Weather Review}, 105:\penalty0 803--816, 1977.

\bibitem[Ni et~al.(2019)Ni, Charoenphakdee, Honda, and Sugiyama]{NCHS19}
C.~Ni, N.~Charoenphakdee, J.~Honda, and M.~Sugiyama.
\newblock On the calibration of multiclass classification with rejection.
\newblock In \emph{Advances in Neural Information Processing Systems},
  volume~32, 2019.

\bibitem[Read et~al.(2011)Read, Pfahringer, Holmes, and Frank]{RPHT11}
J.~Read, B.~Pfahringer, G.~Holmes, and E.~Frank.
\newblock Classifier chains for multi-label classification.
\newblock \emph{Machine Learning}, 85, 2011.

\bibitem[Tarekegn et~al.(2021)Tarekegn, Giacobini, and Michalak]{TGM21}
A.~N. Tarekegn, M.~Giacobini, and K.~Michalak.
\newblock A review of methods for imbalanced multi-label classification.
\newblock \emph{Pattern Recognition}, 118:\penalty0 107965, 2021.

\bibitem[Tharwat(2020)]{Tha20}
A.~Tharwat.
\newblock Classification assessment methods.
\newblock \emph{Applied Computing and Informatics}, 17:\penalty0 168--192,
  2020.

\bibitem[Tsoumakas and Katakis(2007)]{TK07}
G.~Tsoumakas and I.~Katakis.
\newblock Multi-label classification: {A}n overview.
\newblock \emph{International Journal of Data Warehousing and Mining (IJDWM)},
  3:\penalty0 1--13, 2007.

\bibitem[Vaicenavicius et~al.(2019)Vaicenavicius, Widmann, Carl, Lindsten,
  Roll, and Schön]{VWA+19}
J.~Vaicenavicius, D.~Widmann, A.~Carl, F.~Lindsten, J.~Roll, and T.~B. Schön.
\newblock Evaluating model calibration in classification.
\newblock In \emph{Proceedings of the 22nd International Conference on
  Artificial Intelligence and Statistics (AISTATS)}, 2019.

\bibitem[Yang and Koyejo(2020)]{YK20}
F.~Yang and S.~Koyejo.
\newblock On the consistency of top-k surrogate losses.
\newblock In \emph{Proceedings of the 37th International Conference on Machine
  Learning}, pages 10727--10735, 2020.

\bibitem[Zhang and Zhou(2006)]{ZZ06}
M.-L. Zhang and Z.-H. Zhou.
\newblock Multilabel neural networks with applications to functional genomics
  and text categorization.
\newblock \emph{IEEE Transactions on Knowledge and Data Engineering},
  18:\penalty0 1338--1351, 2006.

\bibitem[Zhang and Zhou(2014)]{ZZ14}
M.-L. Zhang and Z.-H. Zhou.
\newblock A review on multi-label learning algorithms.
\newblock \emph{IEEE Transactions on Knowledge and Data Engineering},
  26:\penalty0 1819--1837, 2014.

\end{thebibliography}

\end{document}